\documentclass{llncs}
\usepackage{algorithm}
\usepackage{algpseudocode}
\usepackage{amsmath}
\usepackage{amssymb}
\usepackage{amsfonts}
\usepackage{color}
\allowdisplaybreaks

\usepackage{kantlipsum}

\DeclareMathAlphabet\mathbfcal{OMS}{cmsy}{b}{n}

\newcommand{\ignore}[1]{}
\renewcommand{\Pr}{{\bf Pr}}
\newcommand{\E}{{\bf E}}

\newcommand{\bz}{\boldsymbol{z}}
\newcommand{\bd}{{\boldsymbol{d}}}
\newcommand{\bxi}{{\boldsymbol{\xi}}}
\newcommand{\bI}{{\boldsymbol{I}}}

\begin{document}
\title{A Tight Lower Bound of $\Omega(\log n)$ for the Estimation of the Number of Defective Items}
\author{
    {\bfseries Nader H. Bshouty} \\
    Dept. of Computer Science \\
    Technion\\ Haifa, Israel \\
\texttt{bshouty@cs.technion.ac.il}
\\ \ \\
    {\bfseries Gergely Harcos} \\
        Number Theory Divison\\
Alfréd Rényi Institute of Mathematics\\
    Budapest, Hungary \\
\texttt{harcos.gergely@renyi.hu}
}
\institute{}
\maketitle

\begin{abstract}
Let $X$ be a set of items of size $n$ , which may contain some defective items denoted by $I$, where $I \subseteq X$. In group testing, a {\it test} refers to a subset of items $Q \subset X$. The test outcome is $1$ (positive) if $Q$ contains at least one defective item, i.e., $Q\cap I \neq \emptyset$, and $0$ (negative) otherwise.

We give a novel approach to obtaining tight lower bounds in non-adaptive randomized group testing. Employing this new method, we can prove the following result.

Any non-adaptive randomized algorithm that, for any set of defective items $I$, with probability at least $2/3$, returns an estimate of the number of defective items $|I|$ to within a constant factor requires at least 
$\Omega({\log n})$ tests. 

Our result matches the upper bound of $O(\log n)$ and solves the open problem posed by Damaschke and Sheikh Muhammad in~\cite{DamaschkeM10,DamaschkeM10b} and by Bshouty in~\cite{Bshouty23}.
\end{abstract}

\section{Introduction}

Let $X$ be a set of $n$ items, among which are defective items denoted by $I \subseteq X$. In the context of group testing, a {\it test} is a subset $Q\subseteq X$ of items, and its result is $1$ if $Q$ contains at least one defective item (i.e., $Q\cap I\not=\emptyset$), and $0$ otherwise.

Although initially devised as a cost-effective way to conduct mass blood testing \cite{D43}, group testing has since been shown to have a broad range of applications. These include DNA library screening \cite{ND00}, quality control in product testing \cite{SG59}, file searching in storage systems \cite{KS64}, sequential screening of experimental variables \cite{L62}, efficient contention resolution algorithms for multiple-access communication \cite{KS64,W85}, data compression \cite{HL02}, and computation in the data stream model \cite{CM05}. Additional information about the history and diverse uses of group testing can be found in \cite{Ci13,DH00,DH06,H72,MP04,ND00} and their respective references.

{\it Adaptive} algorithms in group testing employ tests that rely on the outcomes of previous tests, whereas {\it non-adaptive} algorithms use tests independent of the outcome of previous tests\footnote{A test may depend on previous tests but not on the outcomes of the previous tests.}, allowing all tests to be conducted simultaneously in a single step. Non-adaptive algorithms are often preferred in various group testing applications~\cite{DH00,DH06}.

Estimating the number of defective items $d:=|I|$ to within a constant factor of $\alpha$ is the problem of identifying an integer $D$ that satisfies $d\le D< \alpha d$. This problem is widely utilized in a variety of applications~\cite{ChenS90,Swallow85,Thompson62,WalterHB80,GatwirthH89}. 

Estimating the number of defective items in a set $X$ has been extensively studied, with previous works including~\cite{BB18,ChengX14,DamaschkeM10,DamaschkeM10b,FalahatgarJOPS16,RonT16}. 
In this paper, we focus specifically on studying this problem in the non-adaptive setting. Bshouty~\cite{Bshouty19} showed that deterministic algorithms require at least $\Omega(n)$ tests to solve this problem. For randomized algorithms,
Damaschke and Sheikh Muhammad ~\cite{DamaschkeM10b} presented a non-adaptive randomized algorithm that makes $O(\log n)$ tests and, with high probability,
returns an integer $D$ such that $D\ge d$ and $\E[D]=O(d)$. Bshouty \cite{Bshouty19} proposed a polynomial time randomized algorithm that makes $O(\log n)$ tests and, with probability at least $2/3$, returns an estimate of the number of defective items within a constant factor. 

As for lower bounds, Damaschke and Sheikh Muhammad \cite{DamaschkeM10b} gave the lower bound of $\Omega(\log n)$; however, this result holds only for algorithms that select each item in each test uniformly and independently with some fixed probability. They conjectured that any randomized algorithm with a constant failure probability also requires $\Omega(\log n)$ tests. Ron and Tsur~\cite{RonT16}\footnote{The lower bound in~\cite{RonT16} pertains to a different model of non-adaptive algorithms, but their technique implies this lower bound.} and independently 
Bshouty \cite{Bshouty19} prove this conjecture up to a factor of $\log\log n$. 
Recently in \cite{Bshouty23}, Bshouty established a lower bound of  $$\Omega\left(\frac{\log n}{(c\log^* n)^{(\log^*n)+1}}\right)$$ tests, where $c$ is a constant and $\log^*n$ is the smallest integer $k$ such that $\log\log \stackrel{k}{\ldots}\log n<2$.
It follows that the lower bound is
$$\Omega\left(\frac{\log n}{\log\log \stackrel{k}{\ldots}\log n}\right)$$
for any constant $k$.

In this paper, we close the gap between the lower and upper bound. We prove
\begin{theorem}\label{TH1} 
    Let $\alpha=1+\Omega(1)$. Any non-adaptive randomized algorithm that, with probability at least $2/3$, $\alpha$-estimates the number of defective items must make at least
    $$\Omega\left(\frac{\log n}{\log\alpha}\right)$$
    tests. 

    In particular, for algorithms that estimate the number of defective items to within a constant factor, the bound is $\Omega(\log n)$.
\end{theorem}

To prove the Theorem, we first consider any algorithm that makes $m=\log n/(c\log \alpha)$ tests, for a sufficiently large constant $c$, and $\alpha$-estimates the number of defective items. Next, we use this algorithm to construct another one that makes $2m$ tests and, when given any pair of sets of defective items where one set is $\alpha$ times the size of the other set, with high probability, can distinguish which set is the larger of the two.
We then use Yao's principle to turn the algorithm to a deterministic algorithm that can do the same for a random pair of such sets. 
The input pairs are generated with a distribution that is uniform over the logarithm of the size $d$ of the smaller set and uniformly distributed over pairs of subsets of $X$ of sizes $d$ and $\alpha d$.

We then employ a central lemma (Lemma~\ref{Main}) in this paper's analysis. This lemma plays a pivotal role in our proof, requiring an innovative approach for its proof. This Lemma implies that if the number of tests is $2m$ then for an input drawn according to the above distribution, with high probability, the test outcomes for both sets are identical, making them indistinguishable. This leads to a contradiction and, as a result, establishes the lower bound of $m = \Omega(\log n / \log \alpha)$.

The paper is organized as follows: The next section introduces some definitions and notations. In Section~\ref{Section3}, we present the main lemma that plays a crucial role in the proof of Theorem~\ref{TH1}. Then in Section~\ref{Section4} we prove Theorem~\ref{TH1}.

\section{Definitions and Notation}\label{section2}
In this section, we introduce some definitions and notation.

We will consider the set of {\it items} $X=[n]=\{1,2,\ldots,n\}$ and the set of {\it defective items} $I\subseteq X$. 
The algorithm is provided with knowledge of $n$ and has access to a test oracle, denoted as ${\cal O}_I$.
The algorithm uses the oracle ${\cal O}_I$ to make a {\it test} $Q\subseteq X$, and the oracle responds with ${\cal O}_I(Q):=1$ if $Q\cap I\not=\emptyset$, and ${\cal O}_I(Q):=0$ otherwise. 

We say that an algorithm ${\cal A}$ $\alpha$-{\it estimates} the number of defective items with probability at least $1-\delta$ if, for every $I\subseteq X$, ${\cal A}$ runs in polynomial time in $n$, makes tests with the oracle ${\cal O}_I$, and with probability at least $1-\delta$, returns an integer ${\cal A}(I)$ such that\footnote{Some papers in the literature provide the following alternative definition: $|I|/\alpha\le {\cal A}(I)\le \alpha |I|$. It is worth noting that this alternative definition is equivalent to $\alpha^2$-estimation, and the results in this paper also hold for this definition.} $|I|\le {\cal A}(I)< \alpha |I|$. If $\alpha$ is constant, then we say that the algorithm {\it estimates the number of defective items to within a constant factor}. 

The algorithm is called {\it non-adaptive} if the queries are independent of the answers of previous queries and, therefore, can be executed simultaneously in a single step. Our objective is to develop a non-adaptive algorithm that minimizes the number of tests and provides, with a probability of at least $1-\delta$, an $\alpha$ estimation of the number of defective items.

Throughout this paper, all logarithms are taken to the base 2 unless stated otherwise, and bold letters denote random variables.

In the Appendix, we prove the following lemma:
\begin{lemma}\label{Basic1}
    Let ${\cal A}$ be an algorithm that makes $T$ tests and, with probability at least $2/3$, $\alpha$-estimates the number of defective items. Then there is an algorithm ${\cal A}'$ that makes $O(T\log(1/\delta))$ tests and, with probability at least $1-\delta$, $\alpha$-estimates the number of defective items.
\end{lemma}

\section{Preliminary Results}\label{Section3}
In this section, we present the main lemma that plays a crucial role in proving Theorem~\ref{TH1}.

First, we prove the following lemma:
\begin{lemma}\label{main}
Let $n$ be an integer. Given $s$ integers $1=q_1\le q_2\le \cdots\le q_{s-1}\le q_s=n$, define $$\sigma_\ell:=\sum_{i=1}^\ell q_i \mbox{\ \  and\  \ } \tau_\ell:=\sum_{i=\ell+1}^s\frac{1}{q_i}.$$
Then, 
$$\prod_{\ell=1}^{s-1} \max\left(1,\frac{1}{\sigma_\ell\tau_\ell}\right)>\frac{n}{4^s}.$$
\end{lemma}

\begin{proof} 
First, we have
$$\prod_{\ell=1}^{s-1} \left(\frac{q_\ell}{q_{\ell+1}}\frac{\sigma_{\ell+1}}{\sigma_\ell}\frac{\tau_{\ell-1}}{\tau_\ell}\right)=\frac{q_1}{q_s}\cdot \frac{\sigma_s}{\sigma_1}\cdot \frac{\tau_{0}}{\tau_{s-1}}=\sigma_s\tau_0>n.$$
On the other hand, the left-hand side satisfies
\begin{eqnarray*}
  \frac{q_\ell}{q_{\ell+1}}\frac{\sigma_{\ell+1}}{\sigma_\ell}\frac{\tau_{\ell-1}}{\tau_\ell}&=&  \frac{q_\ell}{q_{\ell+1}}\left(1+\frac{q_{\ell+1}}{\sigma_\ell}\right)\left(1+\frac{1}{q_\ell\tau_\ell}\right)\\
  &=& \frac{q_\ell}{q_{\ell+1}}+\frac{q_\ell}{\sigma_\ell}+\frac{1}{q_{\ell+1}\tau_\ell}+\frac{1}{\sigma_\ell\tau_\ell}\\
  &\le&3+\frac{1}{\sigma_\ell\tau_\ell}\le 4\max\left(1,\frac{1}{\sigma_\ell\tau_\ell}\right).
\end{eqnarray*}
Hence $$\prod_{\ell=1}^{s-1} 4\max\left(1,\frac{1}{\sigma_\ell\tau_\ell}\right)>n,$$
and the result follows.\qed
\end{proof}
We now prove the main Lemma. 

\begin{lemma}\label{Main}
    Let $\alpha\ge 2$ and $s=(\log n)/(2000\log \alpha)$. Let $1= q_1\le q_2\le \cdots\le q_s= n$. Let $$Z=\{2^{\lfloor\log\alpha\rfloor+1},2^{\lfloor\log\alpha\rfloor+2},\ldots,2^{\lfloor\log (n/\alpha)\rfloor}\}.$$ Then: 
    $$\Pr_{z\in Z}\left[ \sum_{q_i\le z}q_i\le  \frac{z}{100\alpha} \text{ and } \sum_{q_i\ge z} \frac{1}{q_i}\le \frac{1}{100\alpha z}\right]\ge \frac{99}{100},$$
    where $z$ is uniformly drawn from $Z$. 
\end{lemma}

\begin{proof} 
Let $\sigma_\ell$ and $\tau_\ell$ be as defined in Lemma~\ref{main}.
For each $\ell\in [s-1]$, consider the interval\footnote{If $a>b$ then $[a,b]=\emptyset$.} $I_\ell:=[100\alpha\sigma_\ell,1/(100\alpha\tau_\ell)]$. If $z\in I_\ell$, it satisfies $\sigma_\ell\le z/(100\alpha)$ and $\tau_\ell\le 1/(100\alpha z)$. Additionally, we have $z\ge 100\alpha\sigma_\ell>q_\ell$ and $z\le 1/(100\alpha \tau_\ell)<q_{\ell+1}$. 
Therefore,
$$\sum_{q_i\le z}q_i=\sigma_\ell\le  \frac{z}{100\alpha} \text{ and } \sum_{q_i\ge z} \frac{1}{q_i}=\tau_\ell\le \frac{1}{100 \alpha z}.$$
Furthermore,  $I_\ell\subset (q_\ell,q_{\ell+1}):=\{q|q_\ell<q<q_{\ell+1}\}$. 
As a result, these sets $I_\ell$ are disjoint sets and
therefore
\begin{eqnarray}\label{equiPr}
    \Pr_{z\in Z}\left[ \sum_{q_i\le z}q_i\le  \frac{z}{100\alpha} \text{ and } \sum_{q_i\ge z} \frac{1}{q_i}\le \frac{1}{100\alpha z}\right]\ge \frac{\sum_{\ell=1}^{s-1}|Z\cap I_\ell|}{|Z|}.
\end{eqnarray}

Let $Z'$ be the set of all the powers of 2. We will now show that all the powers of $2$ that are in $I_\ell$ are also in $Z$. That is, $|Z\cap I_\ell|=|Z'\cap I_\ell|$. This follows from two facts. First, the largest powers of $2$ that are in $I:=\cup_\ell I_\ell$ are in $I_{s-1}=[100\alpha\sigma_{s-1},n/(100\alpha)]$, and $\max_{z\in Z} z=2^{\lfloor \log(n/\alpha)\rfloor}>n/(100\alpha)$. Second, the smallest power of $2$ that are in $I$ are in $I_1=[100\alpha,1/(100\alpha\tau_\ell)]$, and $\min_{z\in Z}z=2^{\lfloor\log\alpha\rfloor+1}<100\alpha$.

Using Lemma~\ref{NumT} from the Appendix, the number of powers of $2$ that are in the interval $I_\ell$ is $$|Z'\cap I_\ell|\ge \left\lfloor \log  \max\left(1,\frac{1}{10000\alpha^2\sigma_\ell\tau_\ell}\right)\right\rfloor.$$
    
Therefore, by Lemma~\ref{main}, 
\begin{eqnarray}
\sum_{\ell=1}^{s-1}|Z\cap I_\ell|&=&
\sum_{\ell=1}^{s-1}|Z'\cap I_\ell|\nonumber\\
&\ge&\sum_{\ell=1}^{s-1} \left\lfloor \log  \max\left(1,\frac{1}{10^4\alpha^2\sigma_\ell\tau_\ell}\right)\right\rfloor\nonumber\\ &\ge&  \left(\sum_{\ell=1}^{s-1}   \log  \max\left(1,\frac{1}{10^4\alpha^2\sigma_\ell\tau_\ell}\right)\right)-s\nonumber\\
&=& \log \left(\prod_{\ell=1}^{s-1} \max\left(1,\frac{1}{10^4\alpha^2\sigma_\ell\tau_\ell}\right)\right)-s \nonumber\\&\ge& \log \left(\frac{1}{(10^{4}\alpha^2)^s}\prod_{\ell=1}^{s-1} \max\left(1,\frac{1}{\sigma_\ell\tau_\ell}\right)\right)-s\nonumber\\
&\ge&  \log \left(\prod_{\ell=1}^{s-1} \max\left(1,\frac{1}{\sigma_\ell\tau_\ell}\right)\right)-(15+2\log\alpha)s\nonumber\\&\ge& (\log n-2s)-(15+2\log \alpha)s\nonumber\\
&\ge& \log n-(17+2\log \alpha)\frac{\log n}{2000\log\alpha}\nonumber\\
&\ge& \log n-\frac{19}{2000}\log n\nonumber\\
&\ge& \frac{99}{100} \log n\ge \frac{99}{100} |Z|.\label{kkk}
\end{eqnarray}
By (\ref{equiPr}) and (\ref{kkk}) the result follows.\qed
\end{proof}

\section{The Lower Bound}\label{Section4}

In this section, we present the proof of the theorem that establishes the lower bound on the number of tests required for any non-adaptive randomized algorithm to $\alpha$-estimate the number of defective items, where $\alpha=1+\Omega(1)$.

We prove.

\noindent
{\bf Theorem }\ref{TH1}. 
    {\em Let $\alpha=1+\Omega(1)$. Any non-adaptive randomized algorithm that, with probability at least $2/3$, $\alpha$-estimates the number of defective items must make at least
    $$\Omega\left(\frac{\log n}{\log\alpha}\right)$$
    tests. 

    In particular, for algorithms that estimate the number of defective items to within a constant factor, the bound is $\Omega(\log n)$.}

\begin{proof} 
First, it suffices to prove the lower bound for $\alpha\ge 2$, as any $\alpha$-estimation where $2>\alpha=1+\Omega(1)$ also qualifies as a $2$-estimation, and the lower bound for $2$-estimation is $\Omega(\log n)$, which equates to $\Omega(\log n/\log \alpha)$ when $\alpha=1+\Omega(1)$.

Second, without loss of generality, we assume that $n$ and $\alpha$ are both powers of two. This is because the lower bound for $n'=2^{\lfloor \log n\rfloor}$ and $\alpha'=2^{\lceil \log \alpha \rceil}$ is also a lower bound for $n$ and $\alpha$, and $\Omega(\log n'/\log\alpha')=\Omega(\log n/\log \alpha)$.

Furthermore, we will prove the lower bound for algorithms with a success probability of at least $7/8$. To get a success probability of at least $7/8$, just run the algorithm that has a success probability of at least $2/3$ three times and take the median of the outcomes. See the proof of Lemma~\ref{Basic1}. Therefore, both have the same asymptotic lower bound.

Suppose, to the contrary, that a non-adaptive randomized algorithm ${\cal A}$ exists, which makes
$$s:=\frac{\log n}{2000\log \alpha}$$ tests and, with probability at least $7/8$, $\alpha$-estimates the number of defective items. In other words, for any set of defective items $I\subseteq [n]$, the algorithm ${\cal A}$ makes $s$ random tests (using the oracle ${\cal O}_I$) and, with probability at least $7/8$, returns ${\cal A}(I)$ satisfying $|I|\le {\cal A}(I)< \alpha|I|$. 

Now, we construct an algorithm ${\cal B}$ that, when given two sets of defective items $\{I_0, I_1\}$ where, for some $\xi\in\{0,1\}$, $I_\xi\supset I_{1-\xi}$ and $|I_\xi|= \alpha|I_{1-\xi}|$, makes $2s$ tests (using the oracles ${\cal O}_{I_0}$ and ${\cal O}_{I_1}$), and, with probability at least $3/4$, can determine which of the two sets is larger, effectively outputting $\xi$.

Algorithm ${\cal B}$ first runs algorithm ${\cal A}$ to generate all the tests. This is feasible since algorithm ${\cal A}$ is non-adaptive. Then it makes these tests to both $I_0$ and $I_1$ using ${\cal O}_{I_0}$ and ${\cal O}_{I_1}$, respectively. If ${\cal A}(I_0)>{\cal A}(I_1)$, the algorithm outputs $0$; otherwise, it outputs $1$. 
The probability that neither of the following events occurs: $|I_0|\le {\cal A}(I_0)<\alpha|I_0|$ or $|I_1|\le {\cal A}(I_1)<\alpha|I_1|$, is at most $1/4$. Thus, with probability of at least $3/4$, ${\cal A}(I_\xi)\ge |I_\xi|=\alpha|I_{1-\xi}|>{\cal A}(I_{1-\xi})$, and ${\cal B}$  provides the correct answer.   

We will now define a distribution $D$ over pairs of sets of defective items. Let $D_1$ be the uniform distribution over $N: =\{2^{\log\alpha},2^{\log\alpha+1},\ldots,2^{\log (n/\alpha)-1}\}$. Initially, we select $\bd\in N$ according to the distribution $D_1$. Next, we randomly and uniformly select $\bxi$ from $\{0,1\}$. Finally, we, uniformly at random, draw $\bI_\bxi\subseteq [n]$ of size $\bd$ and $\bI_{1-\bxi}\subseteq [n]$ such that $\bI_{1-\bxi}\supseteq \bI_\bxi$ of size $\alpha \bd$. 

By applying Yao's Principle, we can conclude the existence of a deterministic, non-adaptive algorithm ${\cal C}$ that makes $s$ tests and, when given $\{\bI_0,\bI_1\}$ drawn according to the distribution $D$, with probability of at least $3/4$,  correctly identifies the largest set.

Let $Q_1,Q_2,\ldots,Q_s\subseteq [n]$ be the tests that ${\cal C}$ makes. Note that ${\cal C}$ is deterministic, so $Q_1,Q_2,\ldots,Q_s$ are fixed and non-random. Let $q_i=|Q_i|$ for all $i\in [s]$. We can assume, without loss of generality, that $1= q_1\le q_2\le \cdots\le q_{s-1}\le q_s= n$. In case where $q_1\not=1$ or $q_n\not=n$, then just add the two tests\footnote{The lower bound will then be $s-2$.} $Q_0=\{1\}$ and $Q_{s+1}=[n]$. 

If $\bd\in N$ is drawn according to distribution $D_1$, then $\bz=n/\bd$ is uniformly drawn from $\{2^{\log \alpha+1},2^{\log \alpha+2},\ldots,2^{\log (n/\alpha)}\}$.
By Lemma~\ref{Main}, with probability at least $99/100$, the chosen $\bz=z$ ($\bd=d$) satisfies
\begin{eqnarray}\label{cond1}
    \sum_{q_i\le z} q_i\le \frac{z}{100\alpha }\mbox{\ \ \ and\ \ \ } \sum_{q_i\ge z}\frac{1}{q_i}\le \frac{1}{100\alpha z}.
\end{eqnarray}

Consider $\{\bI_0,\bI_1\}$ drawn according to distribution $D$ conditioned on $\bd=d$ satisfying (\ref{cond1}). Without loss of generality, assume that $|\bI_1|=\alpha d>d=|\bI_0|$. Now let\footnote{$z$ cannot be equal to $q_\ell$ for any $\ell\in[s]$ because, otherwise, $1=q_\ell\cdot ({1}/{q_\ell})\le (\sum_{q_i\le q_\ell}q_i)\sum_{q_i\ge q_\ell}(1/q_i)\le (z/(100\alpha))(1/(100\alpha z)=1/(10^4\alpha^2)<1$.} $q_1\le q_2\le \cdots \le q_\ell< z<q_{\ell+1}\le \cdots\le q_s$. Define the event $A_0$ as the situation where the outcomes of all the tests $Q_1,Q_2,\ldots,Q_\ell$ in algorithm ${\cal C}$ are $0$. Then

\begin{eqnarray} 
\Pr[\neg A_0|\bd=d]&=&
\Pr_{\bI_0,\bI_1,|\bI_0|=d}[(\exists i\in [\ell])({\cal O}_{\bI_0}(Q_i)=1\vee{\cal O}_{\bI_1}(Q_i)=1)]\nonumber\\&=&\Pr_{\bI_0,\bI_1,|\bI_0|=d}\left[\bigvee_{i=1}^\ell\left( \bI_0\cap Q_i\not=\emptyset \vee \bI_1\cap Q_i\not=\emptyset\right)\right]\nonumber\\&=& \Pr_{\bI_1,|\bI_1|=\alpha d}\left[\bigvee_{i=1}^\ell\left(\bI_1\cap Q_i\not=\emptyset\right)\right]\label{A01}\\
&\le &\sum_{i=1}^\ell \Pr_{\bI_1,|\bI_1|=\alpha d}[\bI_1\cap Q_i\not=\emptyset]\label{A02}\\
&=& \sum_{i=1}^\ell \left(1-\prod_{j=0}^{\alpha d-1}\left(1-\frac{q_i}{n-j}\right)\right)\label{A03}\\
&\le& \sum_{i=1}^\ell \left(1-\left(1-\frac{2q_i}{n}\right)^{\alpha d}\right)\label{A04}\\
&\le& \sum_{i=1}^\ell \frac{2\alpha dq_i}{n}=2\alpha \frac{1}{z}\sum_{i=1}^\ell q_i=\frac{1}{50}.\label{A05}
\end{eqnarray} 

(\ref{A01}) follows from the fact that since $\bI_0\subset\bI_1$ we have $\bI_0\cap Q_i\not=\emptyset$ implies $\bI_1\cap Q_i\not=\emptyset$. (\ref{A02}) follows from the union-bound rule. (\ref{A03}) follows from the fact that $\bI_1$ is a random uniform subset of $[n]$ of size $\alpha d$. Therefore, the probability that 
$\bI_1\cap Q_i\not=\emptyset$ is $1-{n-q_i\choose \alpha d}/{n\choose \alpha d}$. Note here that when $n-q_i<\alpha d$ then $1-{n-q_i\choose \alpha d}/{n\choose \alpha d}=1\le (n-q_i+1)q_i/n\le \alpha d q_i/n<2\alpha d q_i/n$ (the term in (\ref{A05})).  In such a case, we can safely disregard the inequality in step (\ref{A04}). Also, for terms where $2q_i/n>1$ we have $\Pr_{\bI_1}[\bI_1\cap Q_i\not=\emptyset]\le 1< \alpha d (2q_i/n) =2\alpha d q_i/n$ and again for those terms you can disregard the inequality in step (\ref{A04}). (\ref{A04}) follows from the fact that $n-j\ge n-\alpha d\ge n-\alpha 2^{\log(n/\alpha)-1}\ge n/2$. (\ref{A05}) follows from the fact that $(1-x)^y\ge 1-yx$ for $x\in [0,1]$ and $y\ge 1$, then from  (\ref{cond1}) and $z=n/d$.

Now define the event $A_1$ as the situation where the outcomes of all the tests $Q_{\ell+1},Q_{\ell+2},\ldots,Q_s$ in algorithm ${\cal C}$ is $1$. Then
\begin{eqnarray}  \Pr[\neg A_1|\bd=d]&=&\Pr_{\bI_0,\bI_1,|\bI_0|=d}[(\exists i\in [\ell])({\cal O}_{\bI_0}(Q_i)=0\vee{\cal O}_{\bI_1}(Q_i)=0)]\nonumber\\ &=&\Pr_{\bI_0,\bI_1,|\bI_0|=d}\left[\bigvee_{i=\ell+1}^s\left( \bI_0\cap Q_i=\emptyset \vee \bI_1\cap Q_i=\emptyset\right)\right]\nonumber\\&=& \Pr_{\bI_0,|\bI_0|=d}\left[\bigvee_{i=\ell+1}^s\left(\bI_0\cap Q_i=\emptyset\right)\right]\label{A16}\\
&\le &\sum_{i=\ell+1}^s \Pr_{\bI_0,|\bI_0|=d}[\bI_0\cap Q_i=\emptyset]\nonumber\\
&=& \sum_{i=\ell+1}^s \left(\prod_{j=0}^{d-1}\left(1-\frac{q_i}{n-j}\right)\right)\nonumber\\
&\le& \sum_{i=\ell+1}^s \left(1-\frac{q_i}{n}\right)^{d}\nonumber\\
&\le &\sum_{i=\ell+1}^s {\frac{n}{dq_i}}=z\sum_{i=\ell+1}^s {\frac{1}{q_i}}\le \frac{1}{100}.\label{A17}
\end{eqnarray}

(\ref{A16}) follows from the fact that $\bI_1\cap Q_i=\emptyset$ implies that $\bI_0\cap Q_i=\emptyset$. (\ref{A17}) follows from the fact that $(1-x)^d\le 1/(dx)$ for any $0<x\le 1$ and $d>0$ combined with (\ref{cond1}) and $\alpha\ge 2$.

Therefore, when considering $\{\bI_0,\bI_1\}$ drawn according to $D$, with probability at least $97/100$ (since $99/100-1/50-1/100=97/100$), algorithm ${\cal C}$ gets the same outcomes for both $\bI_0$ and $\bI_1$. Consequently, the success probability in this case is $1/2$ (essentially guessing). As a result, the overall success probability of ${\cal C}$ cannot be more than $3/100+(1/2)(97/100)=103/200$ which is less than $3/4$. This leads to a contradiction.\qed
\end{proof}
\color{black}

\bibliography{LogNLB}

\begin{thebibliography}{10}

\bibitem{Bshouty19}
Nader~H. Bshouty.
\newblock Lower bound for non-adaptive estimation of the number of defective
  items.
\newblock In Pinyan Lu and Guochuan Zhang, editors, {\em 30th International
  Symposium on Algorithms and Computation, {ISAAC} 2019, December 8-11, 2019,
  Shanghai University of Finance and Economics, Shanghai, China}, volume 149 of
  {\em LIPIcs}, pages 2:1--2:9. Schloss Dagstuhl - Leibniz-Zentrum f{\"{u}}r
  Informatik, 2019.

\bibitem{Bshouty23}
Nader~H. Bshouty.
\newblock Improved lower bound for estimating the number of defective items.
\newblock {\em CoRR}, abs/2308.07721, 2023.

\bibitem{BB18}
Nader~H. Bshouty, Vivian~E. Bshouty{-}Hurani, George Haddad, Thomas Hashem,
  Fadi Khoury, and Omar Sharafy.
\newblock Adaptive group testing algorithms to estimate the number of
  defectives.
\newblock {\em ALT}, 2017.

\bibitem{ChenS90}
Chao~L. Chen and William~H. Swallow.
\newblock Using group testing to estimate a proportion, and to test the
  binomial model.
\newblock {\em Biometrics.}, 46(4):1035--1046, 1990.

\bibitem{ChengX14}
Yongxi Cheng and Yinfeng Xu.
\newblock An efficient {FPRAS} type group testing procedure to approximate the
  number of defectives.
\newblock {\em J. Comb. Optim.}, 27(2):302--314, 2014.

\bibitem{Ci13}
Ferdinando Cicalese.
\newblock {\em Fault-Tolerant Search Algorithms - Reliable Computation with
  Unreliable Information}.
\newblock Monographs in Theoretical Computer Science. An {EATCS} Series.
  Springer, 2013.

\bibitem{CM05}
Graham Cormode and S.~Muthukrishnan.
\newblock What's hot and what's not: tracking most frequent items dynamically.
\newblock {\em {ACM} Trans. Database Syst.}, 30(1):249--278, 2005.

\bibitem{DamaschkeM10}
Peter Damaschke and Azam~Sheikh Muhammad.
\newblock Bounds for nonadaptive group tests to estimate the amount of
  defectives.
\newblock In {\em Combinatorial Optimization and Applications - 4th
  International Conference, {COCOA} 2010, Kailua-Kona, HI, USA, December 18-20,
  2010, Proceedings, Part {II}}, pages 117--130, 2010.

\bibitem{DamaschkeM10b}
Peter Damaschke and Azam~Sheikh Muhammad.
\newblock Competitive group testing and learning hidden vertex covers with
  minimum adaptivity.
\newblock {\em Discrete Math., Alg. and Appl.}, 2(3):291--312, 2010.

\bibitem{D43}
R.~Dorfman.
\newblock The detection of defective members of large populations.
\newblock {\em Ann. Math. Statist.}, pages 436--440, 1943.

\bibitem{DH00}
D.~Du and F.~K Hwang.
\newblock Combinatorial group testing and its applications.
\newblock {\em World Scientific Publishing Company.}, 2000.

\bibitem{DH06}
D.~Du and F.~K Hwang.
\newblock Pooling design and nonadaptive group testing: important tools for dna
  sequencing.
\newblock {\em World Scientific Publishing Company.}, 2006.

\bibitem{FalahatgarJOPS16}
Moein Falahatgar, Ashkan Jafarpour, Alon Orlitsky, Venkatadheeraj Pichapati,
  and Ananda~Theertha Suresh.
\newblock Estimating the number of defectives with group testing.
\newblock In {\em {IEEE} International Symposium on Information Theory, {ISIT}
  2016, Barcelona, Spain, July 10-15, 2016}, pages 1376--1380, 2016.

\bibitem{HL02}
Edwin~S. Hong and Richard~E. Ladner.
\newblock Group testing for image compression.
\newblock {\em {IEEE} Trans. Image Processing}, 11(8):901--911, 2002.

\bibitem{H72}
F.~K. Hwang.
\newblock A method for detecting all defective members in a population by group
  testing.
\newblock {\em Journal of the American Statistical Association},
  67:605–--608, 1972.

\bibitem{KS64}
William~H. Kautz and Richard~C. Singleton.
\newblock Nonrandom binary superimposed codes.
\newblock {\em {IEEE} Trans. Information Theory}, 10(4):363--377, 1964.

\bibitem{GatwirthH89}
Joseph L.Gastwirth and Patricia A.Hammick.
\newblock Estimation of the prevalence of a rare disease, preserving the
  anonymity of the subjects by group testing: application to estimating the
  prevalence of aids antibodies in blood donors.
\newblock {\em Journal of Statistical Planning and Inference.}, 22(1):15--27,
  1989.

\bibitem{L62}
C.~H. Li.
\newblock A sequential method for screening experimental variables.
\newblock {\em J. Amer. Statist. Assoc.}, 57:455--477, 1962.

\bibitem{MP04}
Anthony~J. Macula and Leonard~J. Popyack.
\newblock A group testing method for finding patterns in data.
\newblock {\em Discrete Applied Mathematics}, 144(1-2):149--157, 2004.

\bibitem{ND00}
Hung~Q. Ngo and Ding{-}Zhu Du.
\newblock A survey on combinatorial group testing algorithms with applications
  to {DNA} library screening.
\newblock In {\em Discrete Mathematical Problems with Medical Applications,
  Proceedings of a {DIMACS} Workshop, December 8-10, 1999}, pages 171--182,
  1999.

\bibitem{RonT16}
Dana Ron and Gilad Tsur.
\newblock The power of an example: Hidden set size approximation using group
  queries and conditional sampling.
\newblock {\em {ACM} Trans. Comput. Theory}, 8(4):15:1--15:19, 2016.

\bibitem{SG59}
M.~Sobel and P.~A. Groll.
\newblock Group testing to eliminate efficiently all defectives in a binomial
  sample.
\newblock {\em Bell System Tech. J.}, 38:1179--1252, 1959.

\bibitem{Swallow85}
William~H. Swallow.
\newblock Group testing for estimating infection rates and probabilities of
  disease transmission.
\newblock {\em Phytopathology}, 1985.

\bibitem{Thompson62}
Keith~H. Thompson.
\newblock Estimation of the proportion of vectors in a natural population of
  insects.
\newblock {\em Biometrics}, 18(4):568--578, 1962.

\bibitem{WalterHB80}
S.~D. Walter, S.~W. Hildreth, and B.~J. Beaty.
\newblock Estimation of infection rates in population of organisms using pools
  of variable size.
\newblock {\em Am J Epidemiol.}, 112(1):124--128, 1980.

\bibitem{W85}
Jack~K. Wolf.
\newblock Born again group testing: Multiaccess communications.
\newblock {\em {IEEE} Trans. Information Theory}, 31(2):185--191, 1985.

\end{thebibliography}
\bibliographystyle{plain}

\section*{Appendix}

\noindent
{\bf Lemma \ref{Basic1}.} {\em Let ${\cal A}$ be an algorithm that makes $T$ tests and, with probability at least $2/3$, $\alpha$-estimates the number of defective items. Then there is an algorithm ${\cal A}'$ that makes $O(T\log(1/\delta))$ tests and, with probability at least $1-\delta$, $\alpha$-estimates the number of defective items.}
\begin{proof}
    The algorithm ${\cal A}'$ runs ${\cal A}$ $m=O(\log(1/\delta))$ times ($m$ is odd) and takes the median of the values it outputs. The probability that the median is not in the interval $[|I|,\alpha|I|]$ is the probability that ${\cal A}$ fails at least $\lceil m/2\rceil$ times. By Chernoff's bound, the result follows.\qed
\end{proof}

\begin{lemma}\label{NumT} Let $a,b>0$.
    The number of power of $2$ that are in the interval $[a,b]$ is at least
    $$\left\lfloor\log\max\left(1,\frac{b}{a}\right)\right\rfloor.$$
\end{lemma}
\begin{proof}
If $b<a$ then $[a,b]=\emptyset$ and the number is $0$.

If $b\ge a$ then let $i$ and $j$ be such that $2^i<a\le 2^{i+1}$ and $2^{i+j+1}>b\ge 2^{i+j}$. Then the power of $2$ that are in $[a,b]$ are $\{2^{i+1},2^{i+2},\ldots,2^{i+j}\}$ and their number is $j$. Then
$$j=\log \frac{2^{i+j}}{2^i}> \log \frac{b/2}{a}=\log\frac{b}{a}-1.$$
This implies $j\ge \lfloor \log(b/a)\rfloor$.\qed
\end{proof}

\end{document}